\documentclass[letterpaper, 10 pt, conference]{ieeeconf}  

\IEEEoverridecommandlockouts                              

\usepackage{graphicx}
\usepackage{epsfig} 
\usepackage{mathptmx}
\usepackage{amsmath} 
\usepackage{bm}

\usepackage{amssymb}  
\usepackage{amsthm}
\usepackage{wasysym}
\usepackage[mathcal]{euscript}
\usepackage{bbm}
\usepackage{color}
\usepackage{lipsum}
\usepackage[ruled,vlined,linesnumbered]{algorithm2e}
\usepackage[colorlinks=true,linkcolor=black,anchorcolor=black,citecolor=black,filecolor=black,menucolor=black,runcolor=black,urlcolor=black]{hyperref}
\usepackage{etoolbox}
\usepackage[table,xcdraw,dvipsnames]{xcolor}
\usepackage{multirow}
\usepackage{bbm}

\usepackage{caption}
\usepackage{cleveref}
\Crefname{assum}{Assumption}{Assumptions}
\Crefname{problem}{Problem}{Problems}
\Crefname{thm}{Theorem}{Theorems}
\Crefname{prop}{Proposition}{Propositions}
\Crefname{equation}{}{}

\usepackage{outlines}
\let\labelindent\relax
\usepackage{enumitem}
\setenumerate[2]{label=\alph*.}
\setenumerate[3]{label=\roman*.}

\makeatletter
\patchcmd{\@makecaption}
  {\scshape}
  {}
  {}
  {}
\makeatother

\usepackage[
    style=ieee,
    doi=false,
    isbn=false,
    url=false,
    eprint=false,
    backend=biber,
    natbib=true
    ]{biblatex}
    
\bibliography{references}

\title{\LARGE \bf
\methodname: Training Neural Controllers with \\ Hard, Non-Convex Constraints
}

\author{Long Kiu Chung and Shreyas Kousik
\thanks{All authors are with the Department of Mechanical Engineering, Georgia Institute of Technology, Atlanta, GA.
Corresponding author: \texttt{lchung33@gatech.edu}.
}
}

\begin{document}
\newcommand{\edgar}[1]{{{\textcolor{red}{LKC: #1}}}}
\newcommand{\shrey}[1]{{{\textcolor{red}{SK: #1}}}}

\newtheorem{defn}{Definition}
\newtheorem{rem}[defn]{Remark}
\newtheorem{lem}[defn]{Lemma}
\newtheorem{prop}[defn]{Proposition}
\newtheorem{assum}[defn]{Assumption}
\newtheorem{ex}[defn]{Example}
\newtheorem{thm}[defn]{Theorem}
\newtheorem{cor}[defn]{Corollary}
\newtheorem{problem}[defn]{Problem}

\newcommand{\methodname}{ShardNet\xspace}

\providecommand{\R}{\ensuremath \mathbb{R}}
\newcommand{\N}{\ensuremath \mathbb{N}}

\newcommand{\regtext}[1]{\mathrm{\textnormal{#1}}}
\newcommand{\vc}[1]{\mathbf{#1}}

\newcommand{\st}{\regtext{ s.t. }}

\newcommand{\norm}[1]{\left\Vert#1\right\Vert}
\newcommand{\tp}{^\intercal}
\newcommand{\indicator}{\mathbbm{1}}

\newcommand{\hpoly}{\mathcal{H}}

\newcommand{\ndim}{n}
\newcommand{\ncon}{\ndim_{c}}
\newcommand{\nctrl}{\ndim_{\ctrls}}
\newcommand{\nx}{\ndim_{\xs}}
\newcommand{\ny}{\ndim_{\ys}}
\newcommand{\ndepth}{\ndim_{\depth}}
\newcommand{\npwa}{\ndim_{\regtext{PWA}}}
\newcommand{\npoly}{\ndim_{\ppoly}}
\newcommand{\nFI}{\ndim_{\FIlbl}}
\newcommand{\nsample}{\ndim_{\regtext{sample}}}
\newcommand{\nconvio}{\ndim_{\conlbl}}
\newcommand{\ninvvio}{\ndim_{\invlbl}}

\newcommand{\ppoly}{P}
\newcommand{\ppolyx}{\ppoly_{\xs}}
\newcommand{\ppolyy}{\ppoly_{\ys}}
\newcommand{\ppolyxu}{\ppoly_{\xs\ctrls}}
\newcommand{\Xset}{X}
\newcommand{\Xsetsample}{{\Xset}}
\newcommand{\Vsetsample}{{V}}
\newcommand{\FIset}{\Xset_{\FIlbl}}
\newcommand{\ctrlset}{U}
\newcommand{\paramset}{\Theta}
\newcommand{\classparamset}{\paramset_{\regtext{c}}}
\newcommand{\policyparamset}{\paramset_{\pi}}
\newcommand{\qparamset}{\paramset_{Q}}
\newcommand{\idxset}{I}

\newcommand{\eye}{\vc{I}}
\newcommand{\Acon}{\vc{A}}
\newcommand{\Cpwa}{\vc{C}}
\newcommand{\weight}{\vc{W}}

\newcommand{\zeros}{\vc{0}}
\newcommand{\ones}{\vc{1}}
\newcommand{\bcon}{\vc{b}}
\newcommand{\dpwa}{\vc{d}}
\newcommand{\xv}{\vc{\xs}}
\newcommand{\xnow}{\xv_{\ts}}
\newcommand{\xnext}{\xv_{\ts+1}}
\newcommand{\xsample}{{\xv}}
\newcommand{\yv}{\vc{\ys}}
\newcommand{\uv}{\vc{u}}
\newcommand{\optu}{\uv^{*}}
\newcommand{\exiu}{\uv'}
\newcommand{\unow}{\uv_{\ts}}
\newcommand{\vv}{\vc{v}}
\newcommand{\vsample}{{\vv}}
\newcommand{\vvlb}{\underline{\vv}}
\newcommand{\vvub}{\overline{\vv}}
\newcommand{\zv}{\vc{\zs}}
\newcommand{\bias}{\vc{w}}
\newcommand{\activev}{\vc{a}}
\newcommand{\param}{\vc{\theta}}
\newcommand{\classparam}{\param_{\regtext{c}}}
\newcommand{\policyparam}{\param_{\pi}}
\newcommand{\qparam}{\param_{Q}}
\newcommand{\optctrlparam}{\param_{\ctrls}}

\newcommand{\ts}{t}
\newcommand{\is}{i}
\newcommand{\isopt}{\is^{*}}
\newcommand{\js}{j}
\newcommand{\ks}{k}
\newcommand{\ctrls}{u}
\newcommand{\xs}{x}
\newcommand{\ys}{y}
\newcommand{\zs}{z}
\newcommand{\depth}{d}
\newcommand{\actives}{a}
\newcommand{\bigM}{m}
\newcommand{\dist}{s}
\newcommand{\normdist}{\hat{s}}
\newcommand{\distsample}{{\dist}}
\newcommand{\logscale}{\sigma}
\newcommand{\logmean}{\mu}
\newcommand{\eps}{\epsilon}
\newcommand{\loss}{\ell}

\newcommand{\ffun}{\vc{f}}
\newcommand{\confun}{h}
\newcommand{\pwafun}{\vc{\psi}}
\newcommand{\slicefun}{\regtext{slice}}
\newcommand{\projfun}{\regtext{proj}}
\newcommand{\distfun}{\regtext{dist}}
\newcommand{\normfun}{\regtext{normalize}}
\newcommand{\hardnetcvx}{\regtext{HardNet-Cvx}}
\newcommand{\hardnetnoncvx}{\regtext{\methodname}}
\newcommand{\cardinality}{\#}

\newcommand{\nn}{\vc{\xi}}
\newcommand{\relunn}{\tilde{\xi}}
\newcommand{\basenn}{\nn_{\regtext{b}}}
\newcommand{\classnn}{\nn_{\regtext{c}}}
\newcommand{\policynn}{\nn_{\pi}}
\newcommand{\qnn}{\relunn_{Q}}
\newcommand{\optctrlnn}{\relunn_{\ctrls}}

\newcommand{\FIlbl}{\regtext{FI}}
\newcommand{\conlbl}{\regtext{con}}
\newcommand{\invlbl}{\regtext{inv}}
\newcommand{\distlbl}{\regtext{dist}}
\newcommand{\nonemptylbl}{\neq\emptyset}

\maketitle

\begin{abstract}
While neural network control policies are powerful, their deployment on safety critical systems depends on ensuring that they obey strict constraints.
Existing work often treats safety as a metric to optimize for, which competes with other performance objectives, if training converges at all.
Instead, we introduce \methodname, a neural network architecture that strictly enforces unions of polyhedral constraints \textit{by construction}, using a differentiable projection layer parameterized by a classification network.
The key insight is to embed safety into the neural network's structure, allowing performance to be optimized independently because formal safety guarantees are \textit{always} given.
In contrast with existing neural architectures that can only enforce simple convex constraints, \methodname enables the first safe-by-construction synthesis of forward-invariant neural network controllers on closed-loop systems where safety constraints are expressed as nonconvex unions of polyhedras or learned value function level sets.
To support this, we also introduce a technique to verify and train such value functions \textit{correctly} as rectified linear unit (ReLU) networks, which has not previously been possible.
On double integrator benchmarks drawn from the literature, \methodname policies maintain 100\% safety on verified sets and achieves significantly lower objective loss compared to existing formal methods.
Furthermore, our value function training technique also produces safe sets more than 3 times larger than existing verification approaches.
\end{abstract}
\section{Introduction}

As autonomous dynamical systems are increasingly deployed in everyday tasks, ensuring their safety has become a critical priority.
A common way to certify safety is via the use of control invariant sets, which guarantees \textit{ad infinitum} satisfaction of safety constraints when the system remains within the sets.
However, when neural networks are used to represent control invariant sets \cite{bansal2021deepreach, fisac2019bridging, dawson2023safe, dawson2022learning, dawson2022safe} or learn control policies to stay within these sets (i.e.\ forward invariance) \cite{chung2025provably, harapanahalli2024certified, li2025verifiable}, existing methods either overlook their approximation errors \cite{chung2024goal}, rely on training convergence to balance between safety and performance \cite{chung2025provably, harapanahalli2024certified, li2025verifiable}, or cannot enforce complex, non-convex safety constraints that are common in many systems \cite{min2024hardnet}.

In this work, we present \methodname, a neural network architecture that enables neural controllers to be trained under hard, non-convex safety constraints, while remaining safe throughout the entire training process.
This allows performance objectives to be optimized independently of safety and, to the best of our knowledge, provides the \textit{first} safe-by-construction approach for training neural controllers under unions of input-dependent polyhedral constraints.
We achieve this by combining a classifier with a differentiable projection layer, restricting the output states to lie within safety constraints, specified either explicitly as unions of polyhedras or implicitly as learned value functions.
For the latter case, we additionally provide insights into how they can be learned and verified correctly with rectified linear unit (ReLU) activations, which previous works \cite{li2025verifiable} have failed to accomplish.
Our method does not impose structural assumptions on its base neural controller or on its classifier, and we show it to be effective on piecewise affine (PWA) dynamics and ReLU value functions with 1 hidden layer of width 64.
An overview of \methodname is shown in \Cref{fig:front_figure}.

\begin{figure}[t]
\centering
    \includegraphics[width=1\columnwidth]{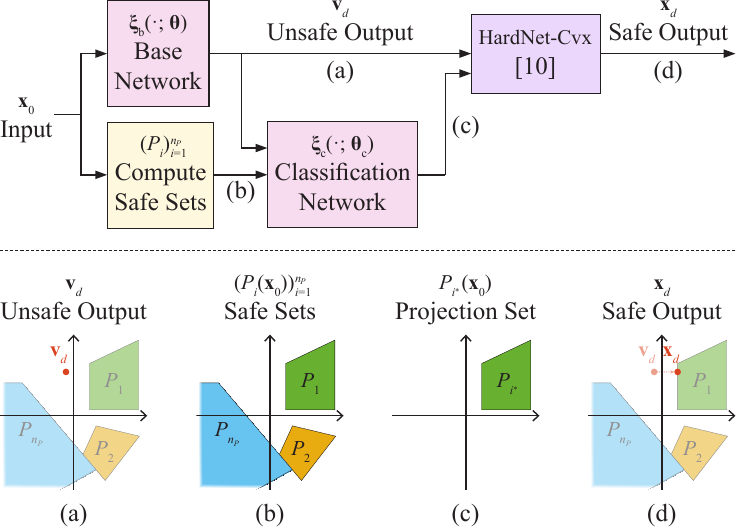}
\caption{
A flowchart of our proposed \methodname method, which restricts the output of a neural network to be within a union of input-dependent polyhedral safe sets \textit{by-construction}.
To accomplish this, we train a classification network $\classnn(\cdot; \classparam)$ to estimate the distances from the violating output $\vv_{\depth}$ of a base network $\basenn(\cdot; \param)$ to each of the safe sets $(\ppoly_\is(\xv_0))_{\is=1}^{\npoly}$, then use HardNet-Cvx \cite{min2024hardnet} to project the output to the closest polyhedral safe set $\ppoly_{\isopt}(\xv_0)$.
}\label{fig:front_figure}
\vspace*{-0.5cm}
\end{figure}

\subsection{Contributions}
\begin{enumerate}
    \item We propose \methodname, a safe-by-construction neural network architecture that enforces unions of input-dependent convex constraints exactly throughout training.
    To the best of our knowledge, this is the \textit{first} neural architecture capable of training under hard non-convex constraints represented as unions of convex sets.
    
    \item We show how \methodname can synthesize forward-invariant neural controllers for safety specifications expressed either as unions of H-polyhedrons or as sub-zero level sets of learned ReLU value functions.
    Compared to neural network repair methods, \methodname guarantees safety independently of training convergence or objective optimization, achieving a 3 times lower objective loss than \cite{chung2025provably}.

    \item We propose a method for training and repairing provably-correct ReLU value functions by augmenting existing sound verification procedures \cite{li2025verifiable} with complete verification checks.
    This significantly reduces unnecessary safe-set shrinkage and yields certified invariant sets 3 times larger than \cite{li2025verifiable}.
\end{enumerate}

\subsection{Paper Organization}
The remainder of the paper is organized as follows: First, we summarize existing works on safe robotic controls and enforcing constraints in neural networks (\Cref{sec:related_work}).
Then, we introduce mathematical concepts used throughout the paper (\Cref{sec:prelim}).
Next, we present the problem and our solution, \methodname, to design a neural network architecture that satisfies a union of convex constraint by-construction (\Cref{sec:hardnet_noncvx}).
We show the utility of \methodname by synthesizing neural control policies for control invariant constraints expressed as unions of H-polyhedrons and sub-zero level sets of ReLU neural networks (\Cref{sec:applications}).
Finally, we demonstrate the validity of our formulation by using \methodname to train forward-invariant neural controllers for double integrators, while comparing our approach's performance against baselines \cite{fisac2019bridging, chung2025provably, li2025verifiable} (\Cref{sec:exp}).
\section{Related Work}\label{sec:related_work}

This work is concerned with designing neural network motion controllers that can enforce control invariance, which is a common strategy for ensuring desiderata such as safety and stability.
To this end, we first review strategies for representing control invariance and learning neural network controllers, then discuss existing works on enforcing safety constraints during neural network training.

\subsection{Control Invariant Set Representations}
Control invariant sets are commonly represented in two ways: intersections of half spaces (H-representation), or sub-zero level sets of value functions.

\subsubsection{Control Invariant Sets from H-Representations}
It is well-known that control invariant sets can be formed by recursively computing the images or preimages of the dynamics \cite{bertsekas1972infinite, blanchini2008set}.
For systems with linear, affine, or PWA dynamics, these sets admit closed-form representations as unions of H or AH-polyhedrons (convex shapes in half space representations and their affine transformations) \cite{anevlavis2019computing, herceg2013multi, vincent2025reachable}.
Then, a forward invariance controller is typically derived via optimization in a model predictive control (MPC) framework \cite{dorea2020robust, houska2024polyhedral}, which can be computationally prohibitive for systems with complex PWA formulations.
To enable fast online inference, recent works have proposed learning forward-invariant neural network controllers \cite{chung2025provably, harapanahalli2024certified}, though successful convergence to a safe solution is not guaranteed.

\subsubsection{Control Invariant Sets from Value Functions}
Alternatively, control-invariant sets can be represented implicitly as the sub-zero level set of a function, most commonly via Hamilton-Jacobi (HJ) reachability analysis or control barrier functions (CBF) (see \cite{choi2021robust, wabersich2023data} for their connections).
Compared to H-representations, these approaches can more naturally encode task constraints (e.g.~collision avoidance for autonomous robots) and extend to other classes of dynamics.
However, traditional HJ methods suffer from the curse of dimensionality \cite{bansal2017hamilton, tomlin2003computational}, while CBF methods imposes restrictive assumptions (such as control-affineness and continuous time) and requires the value function to be handcrafted \cite{ames2019control, ames2014control, ames2016control}.
In response, many methods have resorted to directly \textit{learning} the value functions in HJ \cite{bansal2021deepreach, fisac2019bridging} and CBF \cite{dawson2023safe, dawson2022learning, dawson2022safe} with neural networks.
However, without properly accounting for the approximation errors, learned value functions may forfeit formal safety guarantees \cite{chung2024goal}.
There are a handful of works that have looked into formally repairing (i.e.~training an unsafe neural network until it is safe) and verifying the correctness of learned value functions \cite{li2025verifiable, yang2025scalable}, though their certifications must be associated with a specific control policy, offering no flexibility to find other safe control inputs.

\subsection{Enforcing Constraints in Neural Network Training}

We now review key approaches for enforcing constraints during neural network training.
Our contributions build upon these works.
We first briefly note that a variety of methods exist to repair unsafe controllers with neural networks in the loop \cite{chaudhury2025learning, chung2025provably, harapanahalli2024certified} by extracting training signals from formal verification \cite{chung2021constrained, korda2022stability}, but have no guarantees in being successful, and constraint satisfaction might be at odds with performance metrics.

In contrast, recent work has focused on designing \textit{safe-by-consrtuction} neural network architectures that enforce \textit{hard} constraints such as affineness over polyhedral regions \cite{ataei2025mpolice, balestriero2023police} or outputs being contained within a convex safe set \cite{konstantinov2023new, min2024hardnet}.
These approaches guarantee constraint satisfaction throughout training, allowing performance objectives to be optimized independently.
That said, constraints that commonly arise in dynamics and control are often more complex and non-convex \cite{aceituno2016generalized, fisac2019bridging}: while works such as \cite{lastrucci2025enforce, iftakher2025physics} have attempted to design safe-by-construction networks for general nonlinear constraints, they do so with iterative projections that can only guarantee convergence towards some tolerance near the constraints, so violations can \textit{still} occur.

As such, few \cite{min2024hardnet, bouvier2024policed, wong2026posafenet} have attempted to apply safe-by-construction training methods to problems in control, especially for control invariance.
To the best of our knowledge, the only exceptions are \cite{min2024hardnet, wong2026posafenet}, which trained neural networks to satisfy CBF constraints that are convex on the control inputs assuming control-affine, continuous systems.
Instead, in this paper, we show how \cite{min2024hardnet} can be augmented to \textit{exactly} satisfy a class of non-convex constraints (specifically, a union of polyhedral constraints), which we can express complex forward-invariant conditions in, enabling safe-by-construction control policies in a wider range of problem settings.

\section{Preliminaries}\label{sec:prelim}

This work makes extensive use of H-polyhedrons, neural networks, and PWA functions.
We now introduce our notation and assumptions for each of these concepts.

\subsection{H-Polyhedrons}
An H-polyhedron $\hpoly(\Acon, \bcon) \subseteq \R^\ndim$ is a set of $\ncon \in \N$ half spaces, parameterized by the constraint matrix $\Acon \in \R^{\ncon\times\ndim}$ and the constraint vector $\bcon \in \R^{\ncon}$ as follows \cite{kochenderfer2025algorithms}:
\begin{align}\label{eq:hpoly_def}
    \hpoly(\Acon, \bcon) = \{\zv \mid \Acon \zv \leq \bcon\}.
\end{align}
An H-polyhedron is called an H-polytope if it is bounded.
For a function $\ffun: \R^{\nx} \to \R^{\ny}$, if an $\ny$-dimensional H-polyhedron's parameters ($\Acon$ or $\bcon$) are functions of some input $\xv \in \R^{\nx}$, we denote the H-polyhedron as \textit{input-dependent} \cite{min2024hardnet}.
Furthermore, if $\nexists\, \zv \in \R^{\ndim}$ such that $\Acon\zv\leq\bcon$ (which can be checked with a feasibility linear program (LP) \cite{herceg2013multi}), we denote the H-polyhedron $\hpoly(\Acon, \bcon)$ as empty.

As part of our proposed method, when transforming control-invariant constraints with known states, we reduce certain dimensions in an H-polyhedron to a single value.
This operation is known as \textit{slicing}: per \cite{chung2024goal, herceg2013multi}, we slice the first $\nx < \ndim$ dimensions of an H-polyhedron $\ppoly = \hpoly(\Acon, \bcon) \subseteq \R^{\ndim}$ to a constant $\xv \in \R^{\nx}$ in closed form as:
\begin{align}\label{eq:hpoly_slice}
\begin{split}
    \slicefun(\ppoly_1, \xv) &= \left\{\zv \mid \begin{bmatrix}
        \xv \\
        \zv
    \end{bmatrix}\in\ppoly_1\right\},\\
    &= \hpoly\left((\Acon)_{1:\ncon, (\nx+1):\ndim}, \bcon - (\Acon)_{1:\ncon, 1:\nx}\xv\right).
\end{split}
\end{align}

\subsection{Neural Networks}
A neural network $\nn(\cdot; \param): \R^{\ndim_{0}} \to \R^{\ndepth}$ is a nonlinear function approximator composed from layers of trainable parameters $\param \in \paramset$ and activation functions.
While \methodname does not impose any assumptions on the structure of its base network, we do require some learned components of our method, such as the value function and the optimal control for verification, to be a \textit{ReLU} neural network, which we define below and differentiate with a tilde $\tilde{\cdot}$.
We also briefly summarize the structure of a HardNet-Cvx \cite{min2024hardnet} network, which we primarily built our contributions upon.

\subsubsection{ReLU Neural Networks}
A fully-connected, ReLU-activated feedforward neural network $\relunn(\cdot; \param): \R^{\ndim_{0}} \to \R^{\ndepth}$ returns $\xv_\depth$ given an input $\xv_0$ as:
\begin{subequations}
\begin{align}
    \vv_{\is} &= \weight_{\is}\xv_{\is-1} + \bias_{\is},\label{eq:linear_layer}\\
    \xv_{\is} &= \max\left(\vv_{\is}, \zeros\right),\label{eq:relu_act} \\
    \xv_{\depth} &= \vv_{\depth} = \weight_{\depth}\xv_{\depth-1} + \bias_{\depth},\label{eq:final_layer}
\end{align}
\end{subequations}
where $\is=1,\cdots,\depth-1$ and max is taken elementwise.
We denote $\weight_{1} \in \R^{\ndim_{1}\times \ndim_{0}}, \cdots, \weight_{\depth} \in \R^{\ndim_{\depth}\times \ndim_{\depth - 1}}$ as \textit{weights}, $\bias_{1} \in \R^{\ndim_{1}}, \cdots, \bias_{\depth} \in \R^{\ndim_{\depth}}$ as \textit{biases}, $\depth\in\N$ as the \textit{depth}, $\ndim_{\is}\in\N$ as the \textit{width} of the $\is^{\regtext{th}}$ layer, and layers $1$ to $\depth-1$ as the \textit{hidden layers} of the network.
In this case, the trainable parameters are $\theta = (\weight_1, \cdots, \weight_\depth, \bias_1, \cdots, \bias_\depth)\in\Theta = \R^{\ndim_1\times\ndim_0}\times\cdots\times\R^{\ndim_{\depth}\times \ndim_{\depth - 1}}\times\R^{\ndim_1}\times\cdots\times\R^{\ndepth}$.

It is well known that a ReLU network can be written as a set of mixed-integer linear constraints \cite{tjeng2017evaluating}.
Specifically, \eqref{eq:relu_act} can be written as:
\begin{subequations}\label{eq:relu_to_milp}
\begin{align}
    \xv_{\is} &\geq \vv_{\is},\\
    \xv_{\is} &\leq \vv_{\is} - {\vvlb}_{\is}\tp(1-\activev_{\is}),\\
    \xv_{\is} &\leq {\vvub}_{\is}\tp\activev_{\is},\\
    \xv_{\is} &\geq \zeros\\
    \activev_{\is} &\in \{0, 1\}^{\ndim_{\is}},
\end{align}
\end{subequations}
where the lower bound ${\vvlb}_{\is} \in \R^{\ndim_{\is}}$ and upper bound ${\vvub}_{\is} \in \R^{\ndim_{\is}}$ of $\vv_{\is}$ can be efficiently obtained through verifiers such as CROWN \cite{zhang2018efficient}.
This enables us to embed the constraint $\xv_\depth = \relunn(\xv_0; \param)$ into a mixed-integer linear program (MILP) solver, which we leverage when repairing a learned value function.

\subsubsection{HardNet-Cvx}
Succinctly, HardNet-Cvx \cite{min2024hardnet} adds a differentiable projection layer to a given neural network such that it always obeys a differentiable convex constraint during and after training.
Consider a safe set defined by an input-dependent H-polyhedron that returns $\ppoly(\xv_0) = \hpoly(\Acon(\xv_0), \bcon(\xv_0)) \subseteq \R^{\ndepth}$ for an input $\xv_0 \in \R^{\ndim_0}$ where $\Acon: \R^{\ndim_0} \to \R^{\ncon\times\ndepth}$ and $\bcon: \R^{\ndim_0}\to\R^{\ncon}$ are differentiable, as well as a base neural network $\basenn(\cdot; \param): \R^{\ndim_{0}} \to \R^{\ndepth}$.
Then, HardNet-Cvx returns $\xv_\depth \in \R^{\ndepth}$ as the optimizer of a quadratic program (QP):
\begin{align}\label{eq:hardnet_cvx}
\begin{split}
\xv_\depth &= \hardnetcvx(\xv_0, \ppoly(\xv_0); \param)\\
&= \underset{\xv_\depth}{\arg\min}\{\norm{\xv_\depth - \basenn(\xv_0; \param)}_2 \mid \Acon(\xv_0) \xv_\depth \leq \bcon(\xv_0)\}.
\end{split}
\end{align}
Thus, the output of HardNet-Cvx always lies within the safe set (i.e.\ $\xv_\depth\in\ppoly(\xv_0)$).
Assuming $\ppoly$ is not empty for all possible inputs to the network, \eqref{eq:hardnet_cvx} is differentiable \cite{agrawal2019differentiable, amos2017optnet} and can therefore be incorporated into training.
Note that the HardNet-Cvx structure accepts any differentiable convex constraints; we show only an H-polyhedral constraint in \eqref{eq:hardnet_cvx} for relevance to our method.
Specifically, we extend HardNet-Cvx to the non-convex case when $\ppoly$ is a \textit{union} of input-dependent H-polyhedrons.

\subsection{PWA Functions}
A PWA function $\pwafun: \Xset \to \R^{\ny}$ with output $\yv$ given an input $\xv\in\Xset\subseteq\R^{\nx}$ is parameterized by $\Acon_\is \in \R^{\ncon\times\nx}, \bcon_\is \in \R^{\ncon}, \Cpwa_\is \in \R^{\ny\times\nx}, \dpwa_\is \in \R^{\ny}$ for $\is=1,\cdots,\npwa$ as:
\begin{align}\label{eq:pwa_sys}
    \pwafun(\xv) = \Cpwa_\is \xv + \dpwa_\is\ \forall\  \xv \in \hpoly(\Acon_\is, \bcon_\is),
\end{align}
where we refer to each H-polyhedron $\hpoly(\Acon_\is, \bcon_\is)$ as a PWA region.
For a PWA function to be well-defined, we require that $\Xset=\bigcup_{\is=1}^{\npwa} \hpoly(\Acon_\is, \bcon_\is)$ and that the output is consistent at all intersections of the PWA regions \cite{chung2024goal, chung2025guaranteed}.
In this paper, we denote PWA systems as discrete-time dynamical systems for which the dynamics $\pwafun: \Xset\times\ctrlset \to \R^{\nx}$ is a PWA function.

Similar to ReLU networks, one can express $\ys=\pwafun(\xv)$ as a MILP constraint \cite{marcucci2019mixed} by writing \eqref{eq:pwa_sys} as
\begin{subequations}\label{eq:pwa_to_milp}
\begin{align}
    \Acon_\is \xv &\leq \bcon_\is + \bigM (1-\actives_{\is}),\\
    \yv &\leq \Cpwa_\is \xv + \dpwa + \bigM(1-\actives_{\is}),\\
    \yv &\geq \Cpwa_\is \xv + \dpwa - \bigM(1-\actives_{\is}),\\
    \actives_\is &\in \{0, 1\},\\
    \sum_{\is=1}^{\npwa}\actives_{\is} &= 1,
\end{align}
\end{subequations}
where $\bigM \in \R_+$ is a sufficiently large number.

Consider a union of H-polyhedrons $\ppoly = \bigcup_{\js=1}^{\npoly}\hpoly(\Acon_{\ppoly, \js}, \bcon_{\ppoly, \js})\subseteq\R^{\ny}$.
Its preimage $\pwafun^{-1}(\ppoly)=\{\xv \mid \pwafun(\xv)\in\ppoly\}\subseteq\Xset$ through the PWA function $\pwafun$ is exactly a union of H-polyhedrons \cite{chung2024goal, chung2025guaranteed}:
\begin{align}\label{eq:preimage_pwa}
    \pwafun^{-1}(\ppoly)=\bigcup_{\is=1}^{{\npwa}} \bigcup_{\js=1}^{\npoly} \hpoly\left(\begin{bmatrix}
        \Acon_{\ppoly, \js} \Cpwa_\is\\
        \Acon_\is
    \end{bmatrix},\begin{bmatrix}
        \bcon_{\ppoly, \js} - \Acon_{\ppoly, \js} \dpwa_\is\\
        \bcon_\is
    \end{bmatrix}\right).
\end{align}
In practice, many of the H-polyhedrons in $\pwafun^{-1}(\ppoly)$ are empty and can be omitted for a more memory-efficient representation.

Importantly, ReLU networks can be exactly written as PWA functions \cite{vincent2025reachable, chung2025guaranteed}, meaning that the preimage $\relunn^{-1}(\cdot; \param)$ of a union of H-polyhedrons through a ReLU network has a closed-form expression as another union of H-polyhedrons.
We leverage this property to formulate the forward invariant constraints in our method.
\section{Proposed Method: \methodname}\label{sec:hardnet_noncvx}
\methodname is the core technique that enables our contributions in this paper.
In this section, we first formalize the problem of designing a neural network architecture that obeys a union of H-polyhedral constraints by-construction, then explain how \methodname solves this problem using a differentiable projection layer characterized by a classification network.
An overview of \methodname is shown in \Cref{fig:front_figure}.

\subsection{Problem Setup}\label{sec:hardnet_noncvx_prob}
We wish to design a neural network structure with differentiability and trainable parameters such that it can only return outputs within a safe set, defined as a union of input-dependent H-polyhedrons.
Formally,
\begin{problem}[Hard-Constrained Neural Network with Input-Dependent, Polyhedral Union Constraints]\label{prob:hardnet_noncvx}
Consider an input-dependent safe set represented as a union of H-polyhedrons:
\begin{align}
    \ppoly(\xv_0) = \bigcup_{\is=1}^{\npoly} \ppoly_\is (\xv_0) = \bigcup_{\is=1}^{\npoly}\hpoly(\Acon_{\is}(\xv_0), \bcon_{\is}(\xv_0))\subseteq\R^{\ndepth},   
\end{align}
for an input $\xv_0 \in \Xset_0 \subseteq \R^{\ndim_0}$.
Design a neural network $\nn(\cdot; \param): \Xset_0 \to \R^{\ndepth}$ with learnable parameters $\param\in\paramset$ that satisfies
\begin{align}
    \nn(\xv_0; \param) \in \ppoly(\xv_0)
\end{align}
for all $\xv_0 \in \Xset_0$ and $\param\in\paramset$.
\end{problem}

For a union of input-dependent polyhedrons, it is possible that some of them might be empty for certain inputs.
As such, we assume at least one of them is always non-empty for all possible inputs ($\exists\, \is \in \{1, \cdots, \npoly\}\st\ppoly_\is(\xv_0)\neq\emptyset\, \forall\, \xv_0\in\Xset_0$) so that it is always possible to return a safe output for any input.

\subsection{Proposed Method}
To design a neural network architecture that only returns outputs within a union of H-polyhedrons, we follow the strategy in HardNet-Cvx \cite{min2024hardnet} by first establishing a base network $\basenn(\cdot; \param): \Xset_0 \to \R^{\ndepth}$ with learnable parameters $\param\in\paramset$:
\begin{align}
    \vv_{\depth} = \basenn(\xv_0; \param).
\end{align}
We do not impose a specific structure on the base network but the output $\vv_{\depth}$ may or may not be within the safe set.

To correct $\vv_{\depth}$ to be within the safe set, we need to decide which H-polyhedron in the polyhedral union it should project onto.
Our strategy is to train a classification network $\classnn(\cdot; \classparam): \Xset_0\times\R^{\ndepth} \to \R^{\npoly}$ with parameters $\classparam\in\classparamset$:
\begin{align}
    \begin{bmatrix}
        \dist_1\\
        \vdots\\
        \dist_{\npoly}
    \end{bmatrix} = \classnn\left(\begin{bmatrix}\xv_0\\\vv_{\depth}\end{bmatrix}; \classparam\right),
\end{align}
where $\dist_\is \in \R$ is the predicted distance from $\vv_{\depth}$ to $\ppoly_{\is}(\xv_0)$.
We train this network offline a modified mean squared error (MSE) loss $\loss_{\distlbl} \in \R_+$:
\begin{align}
    \loss_{\distlbl} &=  \frac{\sum_{\js=1}^{\nsample}\sum_{\is\in\idxset_{\nonemptylbl, \js}}\left(\left(\classnn\left(\begin{bmatrix}\xv_{0,\js}\\\vv_{\depth, \js}\end{bmatrix}; \classparam\right)\right)_{\is} - \distsample_{\is, \js}\right)^2}{\sum_{\js=1}^{\nsample}\cardinality\left(\idxset_{\nonemptylbl, \js}\right)},\\
    \idxset_{\nonemptylbl, \js} &= \{\is\in\{1, \cdots, \npoly\}\ |\ \ppoly_\is(\xsample_{0, \js}) \neq \emptyset\},\\
    \distsample_{\is, \js} &= \min\{\norm{\vv_\depth - \vv}_2 \mid \vv \in \ppoly_\is(\xv_0)\},
\end{align}
where $\cardinality(\cdot)$ is the cardinality of a set, $(\begin{bmatrix}
    \xsample_{0, \js}\tp & \vsample_{\depth, \js}\tp
\end{bmatrix}\tp)_{\js=1}^{\nsample}$ are training features uniformly sampled from some domain $\Xsetsample \times \Vsetsample \subset \Xset_0\times\R^{\ndepth}$, $\idxset_{\nonemptylbl, \js} \subseteq \{1, \cdots, \npoly\}$ is the index set of non-empty $\ppoly_\is$, and $\distsample_{\is, \js}$ is the distance between a point and the H-polyhedron $\ppoly_\is$, which can be computed with a QP.

To isolate cases where $\vv_\depth$ is already in a safe set or when an H-polyhedron is empty, we normalize the output of $\classnn(\cdot; \classparam)$ as $\begin{bmatrix}
    \normdist_1 & \cdots & \normdist_{\npoly}
\end{bmatrix}\tp = \normfun(\begin{bmatrix}
    \dist_1 & \cdots & \dist_{\npoly}
\end{bmatrix}\tp)$, defined as
\begin{align}\label{eq:norm_score}
    \normdist_\is = \begin{cases}
         1 & \regtext{if}\ \ppoly_\is(\xv_0) = \emptyset,\\
         0 & \regtext{if}\ \vv_\depth \in \ppoly_\is(\xv_0),\\
         \frac{1}{1+e^{-\logscale(\dist_\is - \logmean)}} & \regtext{otherwise},
    \end{cases}
\end{align}
where $\logscale\in\R_+$ and $\logmean\in\R$ are constants controlling a logistic function such that $0 < \normdist_\is < 1$ if $\ppoly_\is(\xv_0) \neq \emptyset$ and $\vv_\depth \notin \ppoly_\is(\xv_0)$.
This normalization ensures that the predicted distance would always be minimum if $\vv_\depth$ is already in the H-polyhedron, and maximum if the H-polyhedron is empty.
Note that the condition $\ppoly_\is(\xv_0) = \emptyset$ requires solving an LP.
In practice, instead of solving $\npoly$ LPs at inference time, we can precompute the feasible region of each $\ppoly_\is$ via polyhedral projection \cite{herceg2013multi} to reduce the condition to an inequality check.

Finally, we project $\vv_\depth$ to the H-polyhedron with the minimum normalized predicted distance and return it as our final output:
\begin{subequations}\label{eq:hardnet_noncvx}
\begin{align}
    \begin{split}
        \xv_\depth &= \hardnetnoncvx(\xv_0, \ppoly; \param),\\
        &= \hardnetcvx(\xv_0, \ppoly_{\isopt}; \param),\ \regtext{where}
    \end{split}\\
    \isopt &= \underset{\is}{\arg\min}\left\{\left(\normfun\left(\classnn\left(\begin{bmatrix}\xv_0\\\basenn(\xv_0; \param)\end{bmatrix}; \classparam\right)\right)\right)_\is\right\}.\label{eq:opt_hpoly}
\end{align}
\end{subequations}
Thus, $\vv_\depth$ is not altered if it is already in a safe set, and is never projected onto an empty safe set.
Thus, one can think of \methodname as HardNet-Cvx, but with the projection condition controlled by a classification network, thereby preserving its ability to be differentiated and safe-by-construction.

This pipeline, summarized in \Cref{fig:front_figure}, enables the parameters $\param$ in the base network to be trained while ensuring the output is always in the safe set, which is a union of H-polyhedrons. 
We verify that it solves \Cref{prob:hardnet_noncvx}:
\begin{thm}[\methodname Enforces H-Polyhedral Union Constraints]\label{thm:hardnet_noncvx_correct}
    The output $\xv_\depth$ of \methodname is always within the union of H-polyhedrons $\ppoly$ for all $\param\in\paramset$, $\classparam\in\classparamset$, $\logscale \in \R_+$, $\logmean \in \R$, and $\xv_0 \in \Xset_0$.
\end{thm}
\begin{proof}
    From \Cref{sec:hardnet_noncvx_prob}, we know that at least one H-polyhedron is non-empty for any $\xv_0 \in \Xset_0$.
    From \eqref{eq:norm_score}, the normalized predicted distances for all non-empty H-polyhedrons must be strictly less than $1$.
    As such, $\ppoly_{\isopt}$ from \eqref{eq:hardnet_noncvx} always corresponds to a non-empty H-polyhedron, on which HardNet-Cvx would project to.
\end{proof}
\begin{rem}
    An important consequence of \Cref{thm:hardnet_noncvx_correct} is that \methodname still holds its safety guarantees \textit{regardless} of how poorly $\classnn(\cdot; \classparam)$ is trained.
    Though a poorly trained $\classnn(\cdot; \classparam)$ and non-smooth operations such as $\arg\min$ from \eqref{eq:opt_hpoly} may affect the performance of the overall network, we find such cases to be rare in practice.
    We leave a more thorough investigation of their effects to future work.
\end{rem}
\section{Applications}\label{sec:applications}
While solving \Cref{prob:hardnet_noncvx} generally enables polyhedral union constraint satisfaction on neural networks, it is not immediately obvious, however, how the solution can be applied to synthesize safe control policies: a neural control policy with \methodname architecture imposes constraints on the \textit{controller} output, whereas control invariance constraints are most naturally placed on the output of the \textit{dynamics}.
In this section, we demonstrate the utility of \methodname by showing how control invariance constraints, expressed as a union of H-polyhedrons (\Cref{sec:fi_hpolyunion}) or the sub-zero level set of a ReLU network (\Cref{sec:fi_valuefn}), can be formulated into \Cref{prob:hardnet_noncvx} by computing the preimage of PWA systems.

\subsection{Forward-Invariance from Union of H-Polyhedrons}\label{sec:fi_hpolyunion}
We first discuss how to synthesize forward-invariant neural policies with \methodname for a union of H-polyhedrons, which is a common specification for systems with linear, affine, or PWA dynamics with explicit safety constraints such as actuator limits or state bounds.

\subsubsection{Problem Setup}
For a discrete-time PWA dynamics system, given a union of H-polyhedrons, we wish to design neural network control policies such that the next state will be within the polyhedral union, assuming it can be done.
Formally,
\begin{problem}[Learning Forward-Invariant Controller for Union of H-Polyhedrons]\label{prob:fi_polyunion}
    For a discrete-time PWA system with dynamics $\pwafun: \Xset\times\ctrlset \to \R^{\nx}$,
    where $\Xset\times\ctrlset\subseteq\R^{\nx}\times\R^{\nctrl}$ is a union of H-polyhedrons,
    render the set
    \begin{align}
        \FIset = \bigcup_{\is=1}^{\nFI}\Xset_{\FIlbl, \is}=\bigcup_{\is=1}^{\nFI}\hpoly(\Acon_{\FIlbl, \is}, \bcon_{\FIlbl, \is})\subseteq\Xset
    \end{align}
    forward-invariant by designing a neural network control policy $\policynn(\cdot; \policyparam): \Xset \to \ctrlset$ with learnable parameters $\policyparam\in\policyparamset$.
    That is, ensure that
    \begin{align}
        \pwafun\left(\begin{bmatrix}
            \xnow \\
            \policynn(\xnow; \policyparam)
        \end{bmatrix}\right)\in\FIset,
    \end{align}
    for all $\policyparam\in\policyparamset$ and $\xnow\in\Xset$ where $\exists \unow \in \ctrlset$ such that $\pwafun(\begin{bmatrix}
        \xnow\tp & \unow\tp
    \end{bmatrix}\tp)\in\FIset$.
\end{problem}
If $\FIset$ can indeed be rendered forward-invariant, then the control policy from solving \Cref{prob:fi_polyunion} would maintain the states within $\FIset$ for all future timesteps:
\begin{assum}[$\FIset$ Is A Control-Invariant Set]\label{assum:fi_set}
    We assume that for all $\xnow \in \FIset$, $\exists\, \unow \in \ctrlset$ such that $\pwafun(\begin{bmatrix}
        \xnow\tp & \unow\tp
    \end{bmatrix}\tp)\in\FIset$.
\end{assum}
\begin{rem}
\Cref{assum:fi_set} is non-trivial to fulfill, which is why we and other learned forward-invariant control policies \cite{harapanahalli2024certified, chung2025provably} are limited to simple dynamical systems in practice, even though \Cref{prob:fi_polyunion} theoretically extends to all general PWA systems, including those learned as ReLU or convolutional neural networks (CNN).
This motivates us to develop forward-invariant policies for value functions in \Cref{sec:fi_valuefn}, which slightly alleviates the assumption by enabling the invariant sets to be learned as neural networks.
We leave directly reducing the difficulty of \Cref{assum:fi_set} to future work.
\end{rem}

\subsubsection{Proposed Method}

The difficulty of applying \methodname to \Cref{prob:fi_polyunion} lies in defining the H-polyhedral union constraints $\ppoly$ at the output of the control policy $\policynn(\cdot; \policyparam)$, which is separated from the given constraint set $\FIset$ by a PWA system $\pwafun$.
Na\"ively, one might attempt to set $\ppoly(\xnow)=\ppolyxu$, where
\begin{align}
    \ppolyxu = \bigcup_{\is=1}^{\npoly} \ppoly_{\xs\ctrls, \is} = \pwafun^{-1}(\FIset),
\end{align}
which is a union of H-polyhedrons from \eqref{eq:preimage_pwa}.
However, this is incorrect, as the output of $\policynn(\cdot; \policyparam)$ is $\unow\in\ctrlset$, whereas $\ppolyxu\subseteq\Xset\times\ctrlset$.

That said, since we know the current state $\xnow$ at inference time, we can reduce $\ppolyxu$ to $\ctrlset$ by expressing it as a function of $\xnow$ via slicing:
\begin{align}\label{eq:con_fi_polyunion}
    \ppoly(\xnow) = \slicefun(\ppolyxu, \xnow) = \bigcup_{\is=1}^{\npoly} \slicefun(\ppoly_{\xs\ctrls, \is}, \xnow),
\end{align}
which is still a union of H-polyhedrons according to \eqref{eq:hpoly_slice}.

We show that \Cref{prob:fi_polyunion} can be solved by applying \methodname to \eqref{eq:con_fi_polyunion}:
\begin{thm}[\methodname Control Policy for Union of H-Polyhedrons]\label{thm:ctrl_fi_polyunion}
    The control policy
    \begin{align}\label{eq:opt_ctrl_fi_polyunion}
        \unow = \hardnetnoncvx\left(\xnow, \slicefun\left(\pwafun^{-1}(\FIset), \xnow\right); \policyparam\right)
    \end{align}
    satisfies $\pwafun(\begin{bmatrix}
        \xnow\tp & \unow\tp
    \end{bmatrix}\tp) \in \FIset$ for all learnable parameters $\policyparam\in\policyparamset$ and all $\xnow\in\Xset$ where $\exists\  \uv \in \ctrlset$ such that $\pwafun(\begin{bmatrix}
        \xnow\tp & \uv\tp
    \end{bmatrix}\tp)\in\FIset$.
\end{thm}
\begin{proof}
    By the definition of preimage, $\pwafun^{-1}(\FIset)=\{\begin{bmatrix}
        \xv\tp & \uv\tp
    \end{bmatrix}\tp \mid \pwafun(\begin{bmatrix}
        \xv\tp & \uv\tp
    \end{bmatrix}\tp)\in\FIset\}$.
    Thus, by the definition of slicing in \eqref{eq:hpoly_slice}, we have $\ppoly = \slicefun(\pwafun^{-1}(\FIset), \xnow)=\{\uv\mid\pwafun(\begin{bmatrix}
        \xnow\tp & \uv\tp
    \end{bmatrix}\tp)\in\FIset\}$.
    Since, from \Cref{thm:hardnet_noncvx_correct}, we have $ \unow \in \ppoly$ for all $\policyparam\in\policyparamset$ and $\xnow$ where $\ppoly$ is non-empty, and that the slicing and preimage operations are exact, \Cref{thm:ctrl_fi_polyunion} must be true.
\end{proof}
Since \methodname is guaranteed to return a forward-invariant control for $\xnow$ if it exists, the neural control policy can be optimized while always satisfying safety constraints, which are enforced independently of the performance objective.
Furthermore, if the system starts at a state $\xnow$ where $\nexists\,\unow$ to stay in $\FIset$, this dangerous condition can be quickly verified by checking if $\begin{bmatrix}
    \normdist_1 & \cdots & \normdist_{\npoly}
\end{bmatrix}\tp = \ones$ (whether all safe sets for $\xnow$ is empty), which contrasts sharply with \textit{ad hoc} repair methods that cannot verify if the problem is solvable in the first place \cite{chung2025provably, harapanahalli2024certified}.

\subsection{Forward-Invariance from Learned Value Functions}\label{sec:fi_valuefn}

When a system's dynamics and constraint specification are more complex, such as controlling a quadruped to stand up \cite{fisac2019bridging} or navigating a drone through obstacles \cite{dawson2022safe}, it is typical to represent the control invariant set as the sub-zero level set of a value function.
Here, we discuss how to learn these value functions correctly as a ReLU network, and how to synthesize safe neural polices for them using \methodname.

\subsubsection{Problem Setup}
Given a control invariant set expressed as the sub-zero level set of a ReLU neural network value function, we wish to design neural network control policies such that the next state will remain in the sub-zero level set of the value function.
Formally,
\begin{problem}[Learning Forward-Invariant Controller for Value Functions]\label{prob:fi_qnet}
    Given a ReLU value function $\qnn(\cdot; \qparam): \Xset\times\ctrlset \to \R$, design a neural network policy $\policynn(\cdot; \policyparam): \Xset \to \ctrlset$ with learnable parameters $\policyparam\in\policyparamset$ that are differentiable and satisfy
    \begin{align}
        \qnn\left(\begin{bmatrix}
            \xnow \\ \policynn(\xnow; \policyparam)
        \end{bmatrix}; \qparam\right)\leq 0,
    \end{align}
    for all $\param\in\paramset$ and $\xnow \in \Xset \subset \R^{\nx}$ where $\exists \unow \in \ctrlset \subset \R^{\nctrl}$ such that $\qnn(\begin{bmatrix}
        \xnow\tp & \unow\tp
    \end{bmatrix}\tp; \qparam)\leq0$.
    Assume $\Xset\times\ctrlset$ is a union of H-polytopes.
\end{problem}
This specific form of value function $\qnn(\cdot; \qparam)$ is known as a Q-network \cite{li2025verifiable} or a state-action CBF (SACBF) \cite{he2023state}; we will henceforth refer to it as a Q-network.
Q-networks share close ties HJ value functions \cite{fisac2019bridging} and CBFs \cite{he2023state}, and can be approximated as a neural network using deep reinforcement learning (RL) methods such as soft actor-critic (SAC) \cite{fisac2019bridging, haarnoja2018soft}.

Similar to \Cref{assum:fi_set}, we assume that the given Q-network can actually represent a control-invariant set.
In this case, for a discrete-time deterministic, bounded, and Lipschitz continuous dynamical system and some user-defined constraint value function, the Q-network must fulfill \textit{constraint satisfaction} and \textit{forward invariance} to be considered safe \cite{li2025verifiable}. 
Mathematically,
\begin{assum}[Sublevel Set of Q-Network Denotes Safety]\label{assum:safe_qnet}
    Let the system dynamics $\ffun: \Xset\times\ctrlset\to\Xset$ be discrete-time, deterministic, bounded, and Lipschitz continuous:
    \begin{align}
        \xnext = \ffun\left(\begin{bmatrix}
            \xnow \\
            \unow
        \end{bmatrix}\right).
    \end{align}
    Consider a user-defined constraint function $\confun: \R^{\nx} \to \R$, where constraints are considered violated at the current timestep $\ts$ iff $\confun(\xnow) > 0$.
    We assume the sub-zero level set of a Q-network $\{\xv \mid \qnn(\begin{bmatrix}
            \xv\tp & \uv\tp
        \end{bmatrix}\tp; \qparam) \leq 0\}$ is a safe invariant set \cite{li2025verifiable}, meaning that $\qnn(\cdot; \qparam)$ satisfies:
    \begin{enumerate}
        \item (Constraint Satisfaction).
        If $\qnn(\begin{bmatrix}
            \xv\tp & \uv\tp
        \end{bmatrix}\tp; \qparam) \leq 0$, then
        \begin{align}\label{eq:con_qnet}
            \confun(\xv)\leq0.
        \end{align}
        \item (Forward Invariance).
        If $\qnn(\begin{bmatrix}
            \xv\tp & \uv\tp
        \end{bmatrix}\tp; \qparam) \leq 0$, then there exists $\exiu\in\ctrlset$ such that 
        \begin{align}\label{eq:fi_qnet}
            \qnn\left(\begin{bmatrix}\ffun\left(\begin{bmatrix}\xv \\ \uv\end{bmatrix}\right) \\ {\exiu}\end{bmatrix}; \qparam\right)\leq0.
        \end{align}
    \end{enumerate}
\end{assum}
\noindent Then, the control policy $\policynn(\cdot; \policyparam)$ from solving \Cref{prob:fi_qnet} would ensure that any states that starts within the sub-zero level set of the Q-network would satisfy the constraint $\confun(\xnow) \leq 0$ for all future timesteps $\ts$.

Unfortunately, \Cref{assum:safe_qnet} is not trivial to satisfy either.
Existing methods that train Q-networks do not bother to formally verify \eqref{eq:con_qnet} and \eqref{eq:fi_qnet}, and offer no solution to fix them even if violations are found \cite{fisac2019bridging, he2023state}.
The only relevant work is \cite{li2025verifiable}, though they reported having failed to train a safe ReLU Q-network with a non-empty sub-zero level set even for a double integrator system.
Instead, they have to propose a new network structure (\textit{multiplicative Q-network}) where the preimage can no longer be computed by \eqref{eq:preimage_pwa} and the complexity of the verification becomes a more complex mixed-integer quadratically constrained program (MIQCP) in place of a MILP.

This motivates us to offer preliminary insights on how a safe, non-trivial Q-network can be trained with a ReLU network to enable provable forward-invariance on systems using level sets and \methodname.
Particularly, given a ReLU Q-network, we wish to verify whether it satisfies conditions \eqref{eq:con_qnet} and \eqref{eq:fi_qnet}.
If it does not, we wish to extract violation signal from the verification and repair the network until it is safe.
Formally,
\begin{problem}[Verification and Repair of Q-Network]\label{prob:train_qnet}
    Given a ReLU Q-network $\qnn(\cdot; \qparam): \Xset\times\ctrlset \to \R$, system dynamics $\ffun: \Xset\times\ctrlset\to\Xset$, and constraint value function $\confun: \R^{\nx} \to \R$, we wish to find learnable parameters $\qparam=(\weight_1, \cdots, \weight_\depth, \bias_1, \cdots, \bias_\depth)$ such that $\qnn(\cdot; \qparam)$ satisfies \eqref{eq:con_qnet} and \eqref{eq:fi_qnet}, assuming $\ffun$ and $\confun$ are PWA functions.
    In addition, we wish to maximize the size of its sub-zero level set $\{\xv \mid \qnn(\begin{bmatrix}
            \xv\tp & \uv\tp
        \end{bmatrix}\tp; \qparam) \leq 0\}$.
\end{problem}
\noindent Maximizing the sub-zero level set ensures a large safe invariant set and discourages the trivial solution of the Q-network being positive everywhere.

\begin{rem}
Unlike \Cref{prob:fi_polyunion} and \Cref{prob:fi_qnet}, we are not aiming to design a network that satisfies \eqref{eq:con_qnet} and \eqref{eq:fi_qnet} by-construction in \Cref{prob:train_qnet}, but are instead modifying the parameters of an existing network structure.
This is because \eqref{eq:fi_qnet} incurs a recursive condition that is non-trivial to solve with \methodname or any existing hard-constrained neural network architectures; so, we leave this to future work.
\end{rem}

\subsubsection{Safe Controller Synthesis}
We first solve \Cref{prob:fi_qnet} given \Cref{assum:safe_qnet} (i.e.~assuming \Cref{prob:train_qnet} has been solved).
Surprisingly, we find that \Cref{prob:fi_qnet} can be solved in a very similar way to \Cref{prob:fi_polyunion} using \methodname due to ReLU networks being PWA functions.
The key is to compute the preimage of the negative half space instead of $\FIset$:
\begin{thm}[\methodname Control Policy for Sublevel Sets]\label{thm:ctrl_qnet}
    The control policy
    \begin{align}\label{eq:opt_ctrl_qnet}
        \unow = \hardnetnoncvx\left(\xnow, \slicefun\left(\qnn^{-1}(\hpoly(1, 0); \qparam), \xnow\right); \policyparam\right),
    \end{align}
    satisfies $\qnn(\begin{bmatrix}
        \xnow\tp & \unow\tp
    \end{bmatrix}\tp; \qparam)\leq0$ for all $\policyparam\in\policyparamset$ in \eqref{eq:opt_ctrl_qnet} and all $\xnow \in \Xset$ where $\exists\  \uv \in \ctrlset$ such that $\qnn(\begin{bmatrix}
        \xnow\tp & \uv\tp
    \end{bmatrix}\tp; \qparam)\leq0$.
\end{thm}
\begin{proof}
    The H-polyhedron $\hpoly(1, 0)$ represents the set of all non-positive real numbers $\{\zs\mid\zs\leq0\}$ by the definition in \eqref{eq:hpoly_def}.
    The rest of the proof follows that of \Cref{thm:ctrl_fi_polyunion} except with $\FIset$ replaced with $\hpoly(1, 0)$.
\end{proof}

Unlike \cite{fisac2019bridging, li2025verifiable, yang2025scalable} where the learned Q-network must be tied with a specific control policy, \Cref{thm:ctrl_qnet} allows the control policy to be freely altered with additional objectives while still always satisfying constraints by being within the sub-zero level set of the Q-network.

However, the control policy \eqref{eq:opt_ctrl_qnet} relies heavily on \Cref{assum:safe_qnet}.
This is because the policy projects constraint-violating outputs of the base network to the boundary of the sub-zero level set of the value function, which is where violations of \eqref{eq:con_qnet} and \eqref{eq:fi_qnet} are most likely to occur \cite{li2025verifiable, yang2025scalable}.
In practice, one must either project further away from the boundary of the sub-zero level set \cite{he2023state} (e.g.\ use $\hpoly(1, -\eps)$ instead of $\hpoly(1, 0)$ in \eqref{eq:opt_ctrl_qnet} for some small number $\eps \in \R_+$), or find a way to verify and repair a ReLU Q-network, which we address next.

\subsubsection{Training a Safe Q-Network}

Our pipeline for training a safe Q-network is summarized in \Cref{fig:pseudocode}.
It follows closely to the method proposed in \cite{li2025verifiable}, with the difference in adding complete verification conditions to check the validity of the counterexamples found by the sound verification techniques in \cite{li2025verifiable}.
Our key insight is that, if a safe sample is mistakenly identified as a counterexample through a sound verification method, we should alter the verification method instead of the Q-network, as penalizing a Q-network when unnecessary would result in over-shrinkage of the level set (especially if it is a ReLU network) \cite{li2025verifiable}.

\begin{figure}[t]
\centering
    \includegraphics[width=0.8\columnwidth]{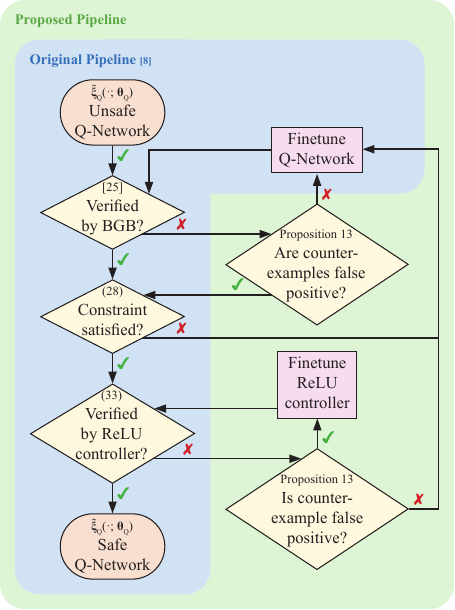}
\caption{
Flowchart for our proposed method (green) of repairing an unsafe Q-network.
The original pipeline proposed by \cite{li2025verifiable} (blue) consists of sound verification methods, which penalizes the Q-network with false positives and causes unnecessary level-set shrinkage.
We augment their method by adding complete verification techniques to double-check whether the identified counterexamples are legitimate.
}\label{fig:pseudocode}
\vspace*{-0.5cm}
\end{figure}

To start, we attempt to find multiple points that violate conditions \eqref{eq:con_qnet} and \eqref{eq:fi_qnet} with boundary-guided backtracking (BGB) \cite{yang2025scalable}.
Succinctly, we collect samples from within the sub-zero level set of the Q-network, then update them by following the gradient to maximize a violation loss of either \eqref{eq:con_qnet} or \eqref{eq:fi_qnet}.
Refer to \cite{yang2025scalable} for details.
After some gradient steps, we retain the samples that violate \eqref{eq:con_qnet} or \eqref{eq:fi_qnet}.
Per \cite{li2025verifiable}, a sample $(\xv_{\conlbl, \is}, \uv_{\conlbl, \is}) \in \{(\xv_{\conlbl, 1}, \uv_{\conlbl, 1}),\cdots,(\xv_{\conlbl, \nconvio}, \uv_{\conlbl, \nconvio})\}$ is \textit{constraint-violating} (violation of \eqref{eq:con_qnet}) if
\begin{subequations}\label{eq:con_violate_check}
\begin{align}
    \qnn\left(\begin{bmatrix}
        \xv_{\conlbl, \is} \\ \uv_{\conlbl, \is}
    \end{bmatrix}; \qparam\right) &\leq 0,\ \regtext{and}\\
    \confun(\xv_{\conlbl, \is}) &> 0,
\end{align}
\end{subequations}
which involves inequality checks.
Similarly, a sample $(\xv_{\invlbl, \is}, \uv_{\invlbl, \is}) \in \{(\xv_{\invlbl, 1}, \uv_{\invlbl, 1}),\cdots, (\xv_{\invlbl, \ninvvio}, \uv_{\invlbl, \ninvvio})\}$ is \textit{invariance-violating} (violation of \eqref{eq:fi_qnet}) if
\begin{subequations} \label{eq:inv_violate_check}
\begin{align}
    \qnn\left(\begin{bmatrix}
        \xv_{\invlbl, \is} \\ \uv_{\invlbl, \is}
    \end{bmatrix}; \qparam\right) &\leq 0, \regtext{and}\label{eq:qnet_says_safe}\\
    \min_{\exiu\in\ctrlset}\left\{\qnn\left(
    \begin{bmatrix}
    \ffun\left(\begin{bmatrix}
        \xv_{\invlbl, \is} \\ \uv_{\invlbl, \is}
    \end{bmatrix}\right)\\\exiu
    \end{bmatrix}; \qparam\right)\right\}&>0.\label{eq:exist_ctrl_check}
\end{align}
\end{subequations}

To check \eqref{eq:exist_ctrl_check}, \cite{li2025verifiable} uses a sound interval arithmetic technique, which means if a solution cannot be found, $(\xv_{\invlbl, \is}, \uv_{\invlbl, \is})$ is \textit{not} an invariance-violating sample; otherwise, the verification is inconclusive. Despite this, \cite{li2025verifiable} would treat potentially false positive samples as if they were true counterexamples, altering $\qparam$ when it might not be needed.
We fix this issue by instead formulating \eqref{eq:exist_ctrl_check} exactly into a MILP check:
\begin{prop}[Forward Invariance MILP Check]\label{prop:inv_check}
Consider some $\xv \in \Xset$ and $\uv \in \ctrlset$.
$(\xv, \uv)$ is an invariance-violating sample iff $\qnn(\begin{bmatrix}
            \xv\tp & \uv\tp
        \end{bmatrix}\tp; \qparam) \leq 0$ and $\optu(\xv, \uv) > 0$, where $\optu:\Xset\times\ctrlset\to\ctrlset$ is the optimal solution of the MILP:
        \begin{align}\label{eq:inv_check}
    \optu(\xv, \uv) = \underset{\exiu\in\ctrlset}{\arg\min}\left\{\qnn\left(
    \begin{bmatrix}
    \ffun\left(\begin{bmatrix}
        \xv \\ \uv
    \end{bmatrix}\right)\\\exiu
    \end{bmatrix}; \qparam\right)\right\}.
\end{align}
\end{prop}
\begin{proof}
    Since we assume $\ffun$ is PWA and $\qnn(\cdot; \qparam)$ is ReLU, \eqref{eq:inv_check} is a MILP from \eqref{eq:relu_to_milp} and \eqref{eq:pwa_to_milp}.
    The rest of the proof follows from the definition of an invariance-violating sample from \eqref{eq:exist_ctrl_check}.
\end{proof}
If true constraint-violating samples and invariance-violating samples are found, per \cite{li2025verifiable}, we train the Q-network $\qnn(\cdot; \qparam)$ for one iteration with the loss function $\loss$ given below in \Cref{eq:FI_learning_loss}, then repeat from BGB \cite{yang2025scalable}:
\begin{subequations}
\begin{align}
    \loss_{\conlbl} =& -\sum_{\is=1}^{\nconvio}\qnn\left(\begin{bmatrix}
        \xv_{\conlbl, \is} \\ \uv_{\conlbl, \is}
    \end{bmatrix}; \qparam\right),\\
    \begin{split}
    \loss_{\invlbl}=& \sum_{\is=1}^{\ninvvio}\Biggl(\qnn\Biggl(
    \begin{bmatrix}
    \ffun\left(\begin{bmatrix}
        \xv_{\invlbl, \is} \\ \uv_{\invlbl, \is}
    \end{bmatrix}\right)\\\optu
    \end{bmatrix}; \qparam\Biggr) \\
    &- \qnn\left(\begin{bmatrix}
        \xv_{\invlbl, \is} \\ \uv_{\invlbl, \is}
    \end{bmatrix}; \qparam\right)\Biggr),\end{split}\\
    \loss =& \frac{\loss_{\conlbl} + \loss_{\invlbl}}{\nconvio + \ninvvio}.\label{eq:FI_learning_loss}
\end{align}
\end{subequations}

If not, we need to formally verify that there exists no $\xv \in \Xset$ and $\uv \in \ctrlset$ that violates \eqref{eq:con_violate_check} and \eqref{eq:inv_violate_check} (i.e.~check every possible point in the domain, not just samples).
For \eqref{eq:con_violate_check}, the condition can be checked using a MILP \cite{li2025verifiable}, which also returns a true counterexample if found.
If \eqref{eq:con_violate_check} is violated, we similarly train $\qnn(\cdot; \qparam)$ for one iteration with $\loss_{\conlbl}$ and repeat from BGB.
Otherwise, we proceed to check violations to \eqref{eq:inv_violate_check} for the whole domain.

Unfortunately, a complete verification of \eqref{eq:inv_violate_check} for all $\xv \in \Xset$ and $\uv \in \ctrlset$ involves a bilevel MILP, which is very difficult to solve \cite{fischetti2017new}.
Instead, \cite{li2025verifiable} proposed a sound verification technique by approximating an optimal controller with a ReLU network $\optctrlnn(\cdot; \optctrlparam):\R^{\nx}\times\R^{\nctrl}\to\R^{\nctrl}$ with supervised learning:
\begin{align}
    \exiu = \optctrlnn\left(\begin{bmatrix}
        \xv\\
        \uv
    \end{bmatrix}; \optctrlparam\right)\approx\optu(\xv, \uv),
\end{align}
where labels can be generated by sampling or solving MILPs.
Then, if the MILP
\begin{align}\label{eq:nn_ctrl_milp}
\begin{split}
    \regtext{find}\ & \xv, \uv\\
    \st &\qnn\left(\begin{bmatrix}
        \xv \\ \uv
    \end{bmatrix}; \qparam\right) \leq 0,\\
    & \qnn\left(\begin{bmatrix}
        \ffun\left(\begin{bmatrix}
            \xv \\ \uv
        \end{bmatrix}\right)\\
        \optctrlnn\left(\begin{bmatrix}
            \xv \\ \uv
        \end{bmatrix}; \optctrlparam\right)
    \end{bmatrix}; \qparam\right)\geq\eps,
\end{split}
\end{align}
is infeasible for some small positive number $\eps\in\R_+$, then we can conclude that $\qnn(\cdot; \qparam)$ is safe and \Cref{prob:train_qnet} has been solved.
Otherwise, the verification is inconclusive, despite \eqref{eq:nn_ctrl_milp} returning a counterexample candidate.
Here, \cite{li2025verifiable} would again ignore false positive cases, penalizing $\qnn(\cdot; \qparam)$ even when it might already be safe.
Instead, we again verify if the candidate is a true invariance-violating counterexample using \Cref{prop:inv_check}.
If it is, we train $\qnn(\cdot; \qparam)$ for one iteration with $\loss_{\invlbl}$ and repeat from BGB.
If not, we instead opt to \textit{tighten the soundness} of the verification by training $\optctrlnn(\cdot; \optctrlparam)$ with more samples and epochs, until either $\qnn(\cdot; \qparam)$ is verified safe, or a true counterexample is found.
A flowchart of the method is summarized in \Cref{fig:pseudocode}.

\begin{rem}
Similar to \cite{li2025verifiable}, we note that many of the techniques used here (such as MILP and sampling) scale poorly with dimensions and thus defeat the purpose of learning an HJ value function in the first place.
Moreover, as with many counterexample neural network repair methods \cite{yang2025scalable, li2025verifiable, tran2020verification, tran2023verification, dai2021lyapunov}, a ``whac-a-mole'' problem might occur \cite{chung2021constrained, chung2025provably}, where making the counterexamples safe might make other samples unsafe, providing no guarantees in the training being successful.
We present this preliminary method in the paper to shine light on how a provably-correct ReLU Q-network can be trained, which has not previously been done, to the best of our knowledge.
We leave investigation of how to repair Q-networks, potentially by set-based training methods \cite{chung2021constrained, chung2025provably, harapanahalli2024certified} or by-construction \cite{konstantinov2023new, min2024hardnet}, to future work.
\end{rem}
\section{Experiments}\label{sec:exp}

We now assess our method by validating \methodname on double integrator benchmarks drawn from the literature \cite{chung2025provably, fisac2019bridging, li2025verifiable}.
All experiments were performed on a desktop computer with a 24-core i9 CPU, 32 GB RAM, and an NVIDIA RTX 4090 GPU.
Our method is implemented in Python.
We are preparing our code for open-source release.

\subsection{Effects of Hard-Constrained Forward Invariance}\label{sec:exp_poly_fi}

To compare safe-by-construction neural networks with neural network repair methods, we replicate the experiment from \cite[Sec. VI-A]{chung2025provably}, where we wish to train a neural network controller to certify forward invariance for a double integrator in a non-convex region.
Our hypothesis is that \methodname creates a more optimal safe control policy, since neural network repair methods need to optimize both safety and performance at the same time, whereas safe-by-construction methods like ours only need to optimize performance since safety is always given.

\subsubsection{Experiment Setup}\label{sec:exp_poly_fi_setup}
We consider double integrator dynamics with $\FIset$ defined as the union of two H-polyhedrons (see \cite[Sec. VI-A.3]{chung2025provably} for details and \Cref{fig:fi_hpoly} for visualization).
The neural network controller is a ReLU network with 1 hidden layer of width 3, and the objective is to approximate the policy $\unow = \begin{bmatrix}
    -2 & -1
\end{bmatrix}\tp \xnow$ without violating the forward invariant constraints.

For \methodname, the classification network has 2 hidden layers of width 20 with $\logscale = 2.5$, $\logmean = 2$ and is trained with $10^5$ samples using the Adam optimizer \cite{kingma2014adam} in PyTorch \cite{paszke2019pytorch} for $10^5$ iterations and learning rate decaying from $0.1$ to $8\times10^{-5}$.
The HardNet-Cvx was trained similarly using $10^3$ samples using Adam \cite{kingma2014adam} and CVXPYlayers \cite{agrawal2019differentiable} for $100$ iterations with a learning rate of $0.1$.
For comparison, we first na\"ively train using only the objective policy with the same hyperparameters as \cite{chung2025provably}.
We then repair the na\"ively trained policy by interleaving repair steps (same setup and hyperparameters as \cite{chung2025provably}) with training steps on the objective policy (same hyperparameters as na\"ive training with learning rate of $10^{-4}$) similar to \cite{chung2021constrained}, in an attempt to further lower their objective loss.
We report the MSE loss with $10^4$ samples of the objective policy within $\FIset$ at each stage for evaluation.

\begin{figure}[t]
\centering
    \includegraphics[width=1\columnwidth]{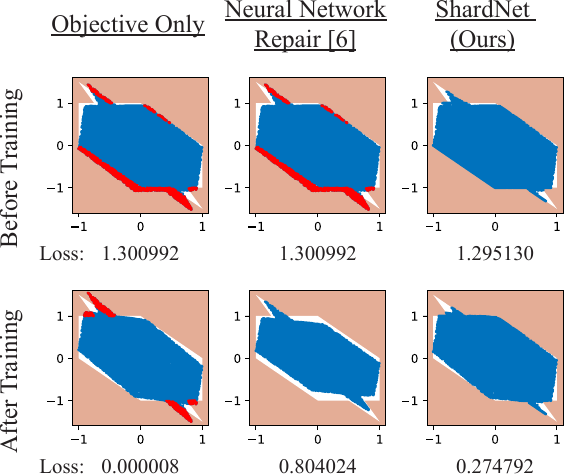}
\caption{
Results of the experiment to certify forward invariance for a double integrator over non-convex safe sets defined by H-polyhedrons.
The complement of the safe states are shown in pink, the states at the next timestep (with current states sampled from the safe sets) using the neural network policy is shown in blue, and the next states violating the forward invariance constraints are highlighted in red.
The objective loss at each stage is listed below each subfigure.
\methodname is safe by construction, allowing the objective to be optimized without needing to balance for safety.
}\label{fig:fi_hpoly}
\vspace*{-0.5cm}
\end{figure}

\subsubsection{Results and Discussion}
The results of this experiment is shown in \Cref{fig:fi_hpoly}.
After na\"ively training, the objective loss decreases to $8\times10^{-6}$, though many samples violated the forward invariance constraints.
The baseline \cite{chung2025provably} successfully repairs the network, but raises the objective loss to $0.8$.
On the other hand, under the \methodname structure, the control policy is automatically safe before training even starts, allowing the method to focus on optimizing the objective loss to $0.3$.
This confirms our hypothesis, showing that \methodname is able to optimize over a union of polyhedral constraints while \textit{always} being safe, whereas repair methods may not always train successfully; even if they do, having to balance safety with the objective at the same time can yield suboptimal performance.

\subsection{Effects of Q-Network Repair Method}\label{sec:exp_q_network}
To assess the quality of our Q-network training method, we compare the ReLU Q-network for a double integrator trained under \cite{li2025verifiable} and our pipeline.
Our hypothesis is that our complete verification method would result in a Q-network with a larger sub-zero level set than \cite{li2025verifiable}.

\subsubsection{Experiment Setup}
We consider the double integrator example from \cite{li2025verifiable} with the same dynamics, domains, and constraints.
We seek to train a ReLU Q-network with 1 hidden layer of width 64 to satisfy \eqref{eq:con_qnet} and \eqref{eq:fi_qnet}.

To begin, we first approximate the Q-network using SAC \cite{fisac2019bridging} with an actor with 2 hidden layers of width 128 using the annealing schedule from \cite{hsu2021safety}, then attempt to repair it using \cite{li2025verifiable} and our pipeline.
For both methods, the hyperparameters follow mostly those from \cite{yang2025scalable}.
The ReLU controller has 2 hidden layers of width 32, and is allowed to train for $6\times10^3$ iterations with Adam \cite{kingma2014adam} and a learning rate of $1\times10^{-4}$ each time before \eqref{eq:nn_ctrl_milp} is evaluated.
For \cite{li2025verifiable} specifically, we use CROWN \cite{zhang2018efficient} to soundly verify the invariance-violating BGB samples.
Finally, we report the number of iterations needed to finetune (whether through BGB \cite{yang2025scalable}, \eqref{eq:con_violate_check}, or \eqref{eq:nn_ctrl_milp}) the Q-network to safety, as well as the area of the sub-zero level set for the verified Q-network, computed by uniformly sampling in a $10^3\times10^3$ grid in the state domain and $512$ points in the control domain.

\subsubsection{Results and Discussion}
The results of this experiment is shown in \Cref{fig:fi_value}.
Surprisingly, \cite{li2025verifiable} produces a non-empty sub-zero level set after $444$ iterations, which is the opposite of what was reported.
This is possibly due to differences in hyperparameters, as \cite{li2025verifiable} never disclosed theirs.
Regardless, our pipeline produces a 3 times larger sub-zero level set in fewer ($113$) iterations for the same network structure.
This result disproves the claim that verifiable Q-networks cannot be trained on ReLU networks \cite{li2025verifiable}, enabling the synthesis of forward-invariant policies on value functions using \methodname, which we show next.

\begin{figure}[t]
\centering
    \includegraphics[width=1\columnwidth]{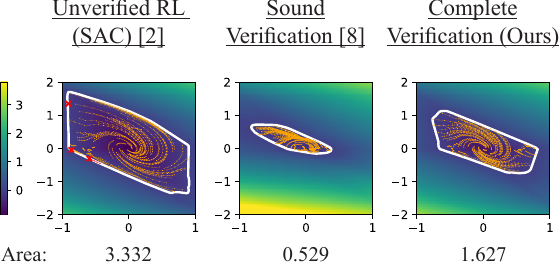}
\caption{
Results of \methodname policies on different ReLU Q-network repairing methods.
The network before repair is shown on the left.
The minimum value output of each state is shown in a yello-blue gradient, with the colorbar displayed on the left.
For each method, we outline the sub-zero level set in white and report their area below.
The rollouts of $100$ trajectories using \methodname policies are shown in orange, with constraint or invariance violating trajectories shown in red (with end-points marked with X).
\methodname do not induce any violations on verified Q-networks, and our Q-network repair method produces a 3 times larger sub-zero level set compared to \cite{li2025verifiable}.
}\label{fig:fi_value}
\vspace*{-0.5cm}
\end{figure}

\subsection{Forward-Invariant Control from Learned Value Function}
Finally, we verify our formulation in \Cref{sec:fi_valuefn} by synthesizing forward-invariant \methodname controllers on the Q-networks trained in \Cref{sec:exp_q_network}.

\subsubsection{Experiment Setup}
For each of the 3 ReLU Q-networks trained in \Cref{sec:exp_q_network}, we first use Reachable Polyhedral Marching (RPM) \cite{vincent2025reachable} to exactly convert them into PWA functions.
Then, we follow the procedure in \Cref{sec:fi_valuefn} by training classification networks with the same setup as \Cref{sec:exp_poly_fi_setup} except with $2 \times 10^4$ iterations, learning rate decaying from $0.1$ to $0.02$, and $3.5 \times 10^3$ samples.
The control policy was also trained to approximate the stabilizing PD control policy $\unow = \begin{bmatrix}
    -2 & -1.5
\end{bmatrix}\tp \xnow$ with the same setup except with $300$ iterations and learning rate decaying from $0.1$ to $0.004$.

For evaluation, we sampled $3 \times 10^3$ points from within the sub-zero level set of the Q-network, then roll out the trajectory for $10$ timesteps using the \methodname policy.
We report the number of violations on the safety constraints and forward invariance, the latter is detected when \methodname cannot return a control that stays within the sub-zero level set, which is caused by a violation of the assumptions in \Cref{prob:hardnet_noncvx}.

\subsubsection{Results and Discussion}
We visualize $100$ of the resulting trajectories in \Cref{fig:fi_value}.
This is a considerably more difficult problem compared with \Cref{sec:exp_poly_fi}, since the number of safe sets ranges from $5,325$ (for \cite{li2025verifiable}) to $13,251$ (for \cite{fisac2019bridging}) instead of just 2.
Despite this, all trajectories from \methodname are still $100\%$ safe on the verified Q-networks trained with \cite{li2025verifiable} and our method.
On the unverified Q-network trained with \cite{fisac2019bridging}, $0.667\%$ of the trajectories violate the safety constraints, and $0.167\%$ of the trajectories violate forward invariance.
The results show that \methodname can still maintain safety on unverified Q-networks, if $\approx1\%$ of failure can be tolerated.
This is especially promising, since Q-networks can be difficult to repair for more complicated systems and scenarios \cite{li2025verifiable}.
\section{Conclusion}

This paper introduces \methodname, a neural network architecture designed to enforce hard constraints, expressed as unions of input-dependent H-polyhedrons, by construction.
By integrating a classification network with a differentiable projection layer, \methodname ensures that the network's output remains within the safe set throughout the training process.
The utility of \methodname is demonstrated by synthesizing forward-invariant control policies for systems with PWA dynamics and learned ReLU value functions.
Experimental results show that \methodname maintains formal safety guarantees without sacrificing performance, achieving lower objective losses compared to traditional neural network repair methods, while ensuring $100\%$ safety on verified Q-networks.

\subsubsection*{Limitations}
We observe 3 major limitations for \methodname.
Firstly, we have only applied \methodname on low-dimensional examples drawn from the literature \cite{chung2025provably, li2025verifiable, fisac2019bridging}.
We are exploring extension to more complex systems.
Secondly, the inference time of \methodname is bottlenecked by solving a differentiable LP projection.
While this issue is inherited from our use of HardNet-Cvx \cite{min2024hardnet}, future work may explore replacing it polyhedral constraints as a union of halfspace constraints, where closed-form projection exists \cite{min2024hardnet}.
Finally, when \methodname is applied on unverified Q-networks, the control policy may fail to return an output when no safe action is possible.
In this case, \methodname needs to be paired with a backup policy to recover safety or to fail gracefully.
Future work may explore the co-design of such a backup policy that is aware of the capabilities of \methodname.

\renewcommand{\bibfont}{\normalfont\footnotesize}
{\renewcommand{\markboth}[2]{}
\printbibliography}

\end{document}